\documentclass[submission,copyright,creativecommons]{eptcs}
\providecommand{\event}{QPL 2014}
\usepackage{breakurl,amssymb,amsmath,amsthm,mathtools,enumerate,etoolbox}
\usepackage{2cptikz}
\usepackage[sort&compress,numbers,comma]{natbib}

\theoremstyle{plain}
\newtheorem{theorem}{Theorem}

\newtheorem{lemma}[theorem]{Lemma}
\newtheorem{proposition}[theorem]{Proposition}

\theoremstyle{definition}
\newtheorem{definition}[theorem]{Definition}

\newcommand{\Two}{\ensuremath{2}}
\newcommand{\CPs}{\ensuremath{\mathrm{CP}^*}}
\newcommand{\cat}[1]{\ensuremath{\mathbf{#1}}}
\newcommand\id[1][]{\ensuremath{\mathrm{id}_{#1}}}
\newcommand{\C}{\ensuremath{\mathbb{C}}}

\newcommand{\bra}[1]{\ensuremath{\left\langle{#1}\right\vert}}
\newcommand{\ket}[1]{\ensuremath{\left\vert{#1}\right\rangle}}

\renewcommand{\-}[0]{\nobreakdash-\hspace{0pt}}
\newcommand\vc[1]{\begin{tabular}{@{}l}#1\end{tabular}}
\newcommand\ignore[1]{}
\newcommand\xto[1]{\ensuremath{\smash{{}\stackrel{\smash{#1}}{\to}{}}}}
\tikzset{every picture/.style={thick}}

\title{Mixed quantum states in higher categories}
\author{
   Chris Heunen \institute{Department of Computer Science,\\University of Oxford} \email{chris.heunen@cs.ox.ac.uk} \thanks{Supported by the Engineering and Physical Sciences Research Council Fellowship EP/L002388/1. \newline We thank Aleks Kissinger for useful discussions.}
   \and Jamie Vicary \institute{
   Centre for Quantum Technologies,\\
   National University of Singapore\\[2pt]
   Department of Computer Science,\\University of Oxford}
   \email{jamie.vicary@cs.ox.ac.uk}
   \and Linde Wester \institute{Department of Computer Science,\\University of Oxford} \email{lindewester@gmail.com} 
}
\def\titlerunning{Mixed quantum states in higher categories}
\def\authorrunning{C. Heunen, J. Vicary \& L. Wester}
\allowdisplaybreaks

\def\ys{0.875}

\begin{document}

\maketitle

\begin{abstract}
  There are two ways to describe the interaction between classical and quantum information categorically: one based on completely positive maps between Frobenius algebras, the other using symmetric monoidal 2-categories. This paper makes a first step towards combining the two. The integrated approach allows a unified description of quantum teleportation and classical encryption in a single 2-category, as well as a universal security proof applicable simultaneously to both scenarios.
\end{abstract}

\section{Introduction}
In the categorical approach to quantum information~\cite{abramskycoecke:cqm}, there are two main approaches to modelling the interaction between classical and quantum data, which can be summarized as follows:
\begin{itemize}
  \item Commutative Frobenius algebras model classical data, noncommutative Frobenius algebras model quantum data, completely positive maps model computational processes. The resulting compact category $\CPs[\cat{FHilb}]$, reviewed in Section~\ref{sec:cpstar}, incorporates both pure and mixed states in a single setting, while admitting a graphical calculus~\cite{selinger:cpm, coeckeperdrix, coeckeheunenkissinger:cpstar, heunenkissingerselinger:cpproj, coeckeheunen:cp, coeckeheunenkissinger:logic,heunencontrerascattaneo, vicary:quantumalgebras}.

  \item Objects model classical data, 1\-morphisms model quantum data, 2\-morphisms model computational processes. The resulting symmetric monoidal 2\-category, reviewed in Section~\ref{sec:hqt}, provides universal syntactic models that can encode entire procedures as single equations~\cite{barvicary, stayvicary, vicary:higherquantumtheory}.
\end{itemize} 
This article makes a first step towards combining both approaches while retaining the advantages of each. Section~\ref{sec:2construction} introduces a procedure that turns a suitable symmetric monoidal category $\cat{C}$ into a symmetric monoidal 2\-category $\Two[\cat{C}]$. It is based on the well-known structure of bimodules and homomorphisms, but with a new definition of bimodule composition in terms of splitting of an idempotent.

Section~\ref{sec:2cpshilb} investigates basic properties of $\Two[\CPs[\cat{FHilb}]]$. We show that on a large domain, which is sufficient for the intended application to quantum information, the 2\-category is well-defined. We also prove the surprising result that every finite groupoid gives rise to an object in $\Two[\CPs[\cat{FHilb}]]$ in a canonical way, suggesting that the 2\-category has a rich structure waiting to be explored.

Finally, Section~\ref{sec:applications} demonstrates the advantages of our combined approach.
We obtain:
\begin{itemize}
  \setlength\itemsep{0pt}
    \item An elegant abstract definition of measurement, that in $\Two[\CPs[\cat{FHilb}]]$ comes down to the usual mixed-state notion of positive operator--valued measurement.
    \item A single equation whose solutions in $\Two[\CPs[\cat{FHilb}]]$ simultaneously include  implementations of quantum teleportation and of classical encrypted communication.
    \item A single proof of security that applies simultaneously to both procedures.
\end{itemize}
There are several interesting directions for future work:
\begin{itemize}
  \setlength\itemsep{0pt}
  \item 
  How can objects of $2[\CPs[\cat{FHilb}]]$ be classified?
  \item 
  Is there a direct construction $\cat{C} \mapsto \mathrm{Mix}[\cat{C}]$ of 2\-categories such that $\Two[\CPs[\cat{C}]] \cong \mathrm{Mix}[2[\cat{C}]]$?
  % \item 
  % Is there a converse construction that recovers $\cat{C}$ from $\Two[\cat{C}]$?
  \item 
  What are nonstandard models such as $\Two[\CPs[\cat{Rel}]]$ like?
  \item 
  Are there nonstandard solutions of the teleportation equation in $\Two[\CPs[\cat{FHilb}]]$, which are neither pure-state quantum teleportation or encrypted communication, but some hybrid process?
\end{itemize}

\subsection{The CP*--construction}
\label{sec:cpstar}

Categorical quantum mechanics deals with dagger monoidal categories~\cite{abramskycoecke:cqm}, which admit a graphical calculus; see~\cite{selinger:graphical}. Within this well-documented setting, let us very briefly recall the CP* construction from the perspective of~\cite[Lemma~1.2]{heunenkissingerselinger:cpproj}; for more details we refer to~\cite{coeckeheunenkissinger:cpstar, heunenkissingerselinger:cpproj, selinger:cpm}. This construction turns a dagger compact category $\cat{C}$ into a new category $\CPs[\cat{C}]$.
% Objects in $\CPs[\cat{C}]$ are normalizable dagger Frobenius algebras $(A,\tinymult[white dot],\tinyunit[white dot])$ in $\cat{C}$. 
An object in $\CPs[\cat{C}]$ is a \emph{special dagger Frobenius algebra} in $\cat{C}$: an object $A$ with morphisms $\tinymult[white dot] \colon A \otimes A \to A$ and $\tinyunit[white dot] \colon I \to A$ satisfying the \emph{specialness} condition $\smash{\tinyhandle[white dot]} = \smash{\tinyid}$, as well as the \emph{dagger Frobenius algebra} laws:
{\def\quad{\hspace{0.1cm}}
\begin{calign}
\label{eq:frobenius}
  \begin{pic}[yscale=\ys]
    \node[white dot] (l) at (0,0) {};
    \node[white dot] (r) at (.5,.5) {};
    \draw (l.center) to[out=90,in=180] (r.center);
    \draw (r.center) to (.5,1);
    \draw (r.center) to[out=0,in=90] (1,-.5);
    \draw (l.center) to[out=180,in=90] (-.5,-.5);
    \draw (l.center) to[out=0,in=90] (.5,-.5);
  \end{pic}
  \quad=\quad
  \begin{pic}[yscale=\ys]
    \node[white dot] (r) at (0,0) {};
    \node[white dot] (l) at (-.5,.5) {};
    \draw (r.center) to[in=0,out=90] (l);
    \draw (l.center) to (-.5,1);
    \draw (l.center) to[in=90,out=180] (-1,-.5);
    \draw (r.center) to[in=90,out=0] (.5,-.5);
    \draw (r.center) to[in=90,out=180] (-.5,-.5);
  \end{pic}
  &
  \begin{pic}[yscale=\ys]
    \node[white dot] (m) at (0,.5) {};
    \node[white dot] (u) at (-.5,0) {};
    \draw (m.center) to (0,1);
    \draw (u.center) to[out=90,in=180] (m.center);
    \draw (m.center) to[out=0,in=90] (.5,-.5);
  \end{pic}
  \quad=\quad
  \begin{pic}[yscale=\ys]
    \draw (0,-.5) to (0,1);
  \end{pic}
  \quad=\quad
  \begin{pic}[yscale=\ys]
    \node[white dot] (m) at (0,.5) {};
    \node[white dot] (u) at (.5,0) {};
    \draw (m.center) to (0,1);
    \draw (u.center) to[out=90,in=0] (m.center);
    \draw (m.center) to[out=180,in=90] (-.5,-.5);
  \end{pic}
  &
  \begin{pic}[xscale=.75, yscale=\ys]
    \node[white dot] (t) at (.5,.5) {};
    \node[white dot] (b) at (-.5,0) {};
    \draw (b.center) to[out=0,in=180] (t.center);
    \draw (-.5,-.5) to (b.center);
    \draw (t.center) to (.5,1);
    \draw (t.center) to[out=0,in=90] (1,-.5);
    \draw (-1,1) to[out=-90,in=180] (b.center);
  \end{pic}
\quad=\quad
  \begin{pic}[xscale=.75, yscale=\ys]
    \node[white dot] (t) at (0,.5) {};
    \node[white dot] (b) at (0,0) {};
    \draw (b.center) to (t.center);
    \draw (t.center) to[out=180,in=-90] (-.5,1);
    \draw (t.center) to[out=0,in=-90] (.5,1);
    \draw (b.center) to[out=180,in=90] (-.5,-.5);
    \draw (b.center) to[out=0,in=90] (.5,-.5);
  \end{pic}
  \quad=\quad
  \begin{pic}[xscale=.75, yscale=\ys]
    \node[white dot] (t) at (-.5,.5) {};
    \node[white dot] (b) at (.5,0) {};
    \draw (b.center) to[out=180,in=0] (t.center);
    \draw (.5,-.5) to (b.center);
    \draw (t.center) to (-.5,1);
    \draw (t.center) to[out=180,in=90] (-1,-.5);
    \draw (1,1) to[out=-90,in=0] (b);
  \end{pic}
\end{calign}}%
Commutative such objects are also called \emph{classical structures}.
A morphism $(A,\tinymult[white dot],\tinyunit[white dot]) \to (B,\tinymult[gray dot],\tinyunit[gray dot])$ in $\CPs[\cat{C}]$ is a morphism $f \colon A \to B$ in $\cat{C}$ satisfying the \textit{complete positivity} condition
\begin{equation}\label{eq:cpstar-condition}
  \begin{pic}[yscale=\ys]
    \node[white dot] (w) at (-.25,-.7) {};
    \node[morphism] (f) at (0,0) {$f$};
    \node[gray dot] (b) at (-.25,.7) {};
    \draw (f.north) to[out=90,in=0] (b);
    \draw (b.center) to[out=90,in=90,looseness=1.5] (.5,.7) to (.5,-1);
    \draw (w.center) to[out=0,in=-90] (f.south);
    \draw (-.25,-1) to (w.center);
    \draw (w.center) to[out=180,in=-90,looseness=1.25] (-1,0) to (-1,1.25);
    \draw (-.75,1.25) to (-.75,.4) to[out=-90,in=-90, looseness=2] (-.5,.4) to[out=90,in=180] (b.center);
  \end{pic}
  \quad = \quad
  \begin{pic}[yscale=\ys]
    \draw (0.3,0.35) to (0.3,1.125);
    \draw (-0.3,0.35) to (-0.3,1.125);
    \draw (0.3,-0.35) to (0.3,-1.125);
    \draw (-0.3,-0.35) to (-0.3,-1.125);
    \draw (0,-0.35) to (0,0.35);
    \node[morphism, minimum width=10mm] (g) at (0,-.35) {$g$};
    \node[morphism, minimum width=10mm] (gd) at (0,.55) {$g^\dag$};
  \end{pic}
  % \begin{pic}
  %   \node[morphism] (f) at (0,0) {$f$};
  %   \node[gray dot] (t) at (0,.66) {};
  %   \node[white dot] (b) at (0,-.66) {};
  %   \draw[<-] (-.5,-1.2) to[out=90,in=-150] (b);
  %   \draw[->] (.5,-1.2) to[out=90,in=-30] (b);
  %   \draw[->] (b) to (f.south);
  %   \draw[->] (f.north) to (t);
  %   \draw[->] (t) to[out=30,in=-90] (.5,1.2);
  %   \draw[<-] (t) to[out=150,in=-90] (-.5,1.2);
  % \end{pic}
  % \quad = \quad
  % \begin{pic} 
  %   \node[morphism, minimum width=10mm] (l) at (-.75,0) {$g_*$};
  %   \node[morphism, minimum width=10mm] (r) at (.75,0) {$g$};
  %   \draw[->] (.75,-1.2) to (r.south);
  %   \draw[<-] (-.75,-1.2) to (l.south);
  %   \draw[->] ([xshift=7pt]r.north) to (1,1.2);
  %   \draw[<-] ([xshift=-7pt]l.north) to (-1,1.2);
  %   \draw[->] ([xshift=-7pt]r.north) to[out=90,in=90,looseness=1.5] ([xshift=7pt]l.north);
  % \end{pic}
% \end{equation}
% for some object $X$ and morphism $g \colon A \to X \otimes B$ of $\cat{C}$. 
\end{equation}
for some morphism $g \colon A \otimes B^* \to X$ in $\cat{C}$. 
This gives a well-defined dagger compact category $\CPs[\cat{C}]$ with the following basic interpretation:
\begin{center}\begin{tabular}{lll}
  \textbf{Category theory} & \textbf{Geometry} & \textbf{Interpretation} \\
  Commutative objects & Lines with commutative dots & Classical information \\
  Noncommutative objects & Lines with noncommutative dots & Quantum information \\
  Morphisms & Vertices & Physical operations
\end{tabular}\end{center}

The CP*--construction is of fundamental importance because it turns a category of pure states and processes into a category of mixed states: applying the CP*--construction to the category $\cat{FHilb}$ of finite-dimensional Hilbert spaces and linear maps results in the category $\CPs[\cat{FHilb}]$ of finite-dimensional C*-algebras and completely positive maps.

\subsection{Higher quantum theory}
\label{sec:hqt}

Higher quantum theory~\cite{vicary:higherquantumtheory,vicary:topology} separates classical and quantum information by replacing monoidal categories by monoidal weak 2\-categories. These also have a graphical notation~\cite{lauda:ambidextrous}: %, giving the following basic interpretation:
\begin{center}\begin{tabular}{lll}
  \textbf{Category theory} & \textbf{Geometry} & \textbf{Interpretation} \\
  Objects & Surfaces & Classical information \\
  1\-Morphisms & Lines & Quantum systems \\
  2\-Morphisms & Vertices & Physical operations 
\end{tabular}\end{center}
Graphically, composition of 1\-morphisms is indicated by horizontal juxtaposition, and composition of 2\-morphisms by vertical juxtaposition. The tensor product is given by `overlaying' regions one above the other, perpendicular to the plane of the page.

Just like in the 1\-categorical case, the diagrams are interpreted as describing sequences of events taking place over time, with time running from bottom to top. A dagger provides a formal time-reversal of 2\-morphisms, represented graphically by reflecting a diagram about a horizontal axis.

\begin{definition}
  A \textit{dagger 2\-category} is a 2\-category equipped with an involutive operation $\dag$ on 2\-morphisms, such that $\mu^\dag \colon G \Rightarrow F$ for all $\mu \colon F \Rightarrow G$, in a way that is functorial and compatible with the rest of the monoidal 2\-category structure.
\end{definition}

The core theory uses the graphical components summarized below, motivated in detail in~\cite{vicary:higherquantumtheory}.
\def\aascale{1.4}
\def\aaspace{\hspace{30pt}}
\setlength\fboxsep{0pt}
\def\sep{5pt}
\def\innersep{3pt}
\def\littlegap{12pt}
\def\boxmargin{0.15cm}
\newcommand{\centerdia}[1]{#1}
\newcommand\newtwocell[2]
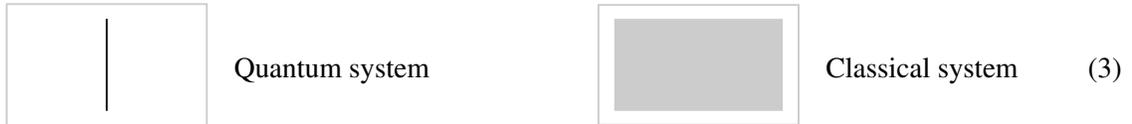
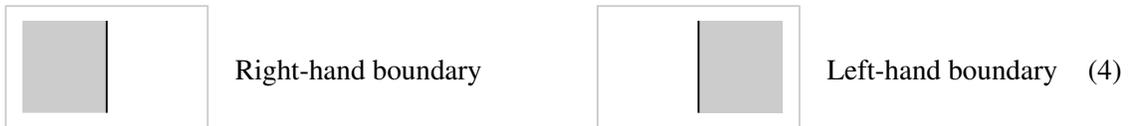
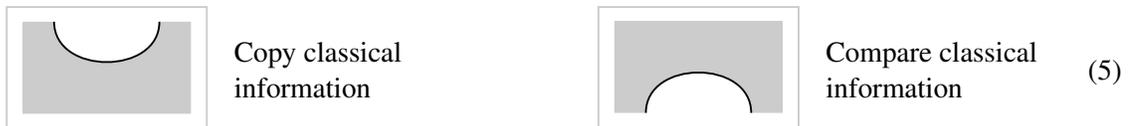
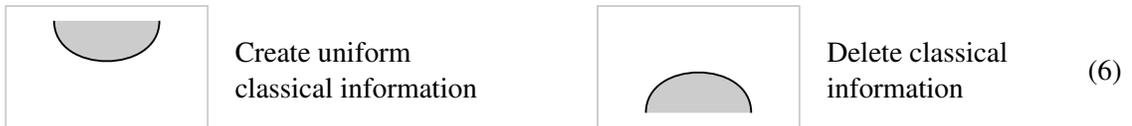
{\begin{aligned}
\begin{tikzpicture}[scale=\aascale, yscale=\ys]
    #1
    \draw [black!20]
        ([xshift=-\boxmargin, yshift=-\boxmargin] current bounding box.south west)
        rectangle
        ([xshift=\boxmargin, yshift=\boxmargin] current bounding box.north east);
\end{tikzpicture}
\end{aligned}
\hspace{10pt} \makebox[80pt][l]{\vc{#2}}}
\newcommand\separatetwocells{\\[\sep]}
\begin{calign}
\newtwocell{
    \draw [fill=white, draw=none] (0.2,-0.5)
        to (0.5,-0.5)
        to (1.8,-0.5)
        to (1.8,-1.5)
        to (0.2,-1.5)
        to (0.2,-0.5);
    \draw [thick] (1,-0.5)
        to (1,-1.5);
}
{Quantum system}
%\separatetwocells
\hspace{\littlegap}&\hspace{\littlegap}
\newtwocell{
    \draw [fill=\fillClight, draw=none] (0.2,0.5)
        to (0.5,0.5)
        to (1.8,0.5)
        to (1.8,1.5)
        to (0.2,1.5)
        to (0.2,0.5);
}
{Classical system}
\separatetwocells
\label{eq:firsttopboundary}
%\nonumber
\newtwocell{
    \draw [white] (1.8,-0.5) rectangle (0.2,-1.5);
    \draw [fill=\fillClight, draw=none] (0.2,-0.5)
        to (0.5,-0.5)
        to (1,-0.5)
        to (1,-1.5)
        to (0.2,-1.5)
        to (0.2,-0.5);
    \draw [thick] (1,-0.5)
        to (1,-1.5);
}
{Right-hand boundary}
%\separatetwocells
\hspace{\littlegap}&\hspace{\littlegap}
\newtwocell{
    \draw [white] (1.8,-0.5) rectangle (0.2,-1.5);
    \draw [fill=\fillClight, draw=none] (1.8,-0.5)
        to (1,-0.5)
        to (1,-1.5)
        to (1.8,-1.5)
        to (1.8,-0.5);
    \draw [thick] (1,-0.5)
        to (1,-1.5);
}
{Left-hand boundary}
\separatetwocells
\label{eq:secondtopboundary}
%\nonumber
\newtwocell{
    \draw [fill=\fillClight, draw=none] (0.2,-0.5)
        to (0.5,-0.5)
        to [out=down, in=down, looseness=1.5] (1.5, -0.5)
        to (1.8,-0.5)
        to (1.8,-1.5)
        to (0.2,-1.5)
        to (0.2,-0.5);
    \draw [thick] (0.5,-0.5)
        to [out=down, in=down, looseness=1.5] (1.5,-0.5);
}
{Copy classical\\information}
%\separatetwocells
\hspace{\littlegap}&\hspace{\littlegap}
\newtwocell{
    \draw [fill=\fillClight, draw=none] (0.2,0.5)
        to (0.5,0.5)
        to [out=up, in=up, looseness=1.5] (1.5,0.5)
        to (1.8,0.5)
        to (1.8,1.5)
        to (0.2,1.5)
        to (0.2,0.5);
    \draw [thick] (0.5,0.5)
        to [out=up, in=up, looseness=1.5] (1.5,0.5);
}
{Compare classical\\information}
\separatetwocells
\label{eq:lasttopboundary}
%\nonumber
\newtwocell{
    \draw [white] (0.2,1) to (1.8,1);
    \draw [fill=\fillClight, thick] (0.5,1.5)
        to [out=down, in=down, looseness=1.5] (1.5,1.5);
    \draw [thick, white] (1,0.5) to (1,1);
}
{Create uniform\\classical information}
%\separatetwocells
\hspace{\littlegap}&\hspace{\littlegap}
\newtwocell{
    \draw [white] (0.2,-1) to (1.8,-1);
    \draw [fill=\fillClight, thick] (0.5,-1.5)
        to [out=up, in=up, looseness=1.5] (1.5,-1.5);
    \draw [thick, white] (1,-0.5) to (1,-1);
}
{Delete classical\\information}
\end{calign}

\begin{definition}\label{def:topologicalboundary}
  An object in a symmetric monoidal 2\-category has a \emph{topological boundary} when it is equipped with data \eqref{eq:firsttopboundary}\-\eqref{eq:lasttopboundary} satisfying the following axioms, which amount to saying that the boundary of a classical system is topological and that holes can be eliminated:
  % equations~\eqref{eq:top1} and \eqref{eq:top2}.
% These components are required to satisfy %the following axioms:
% a set of axioms, which amount to saying that the boundary of a classical system is topological, and that holes can be eliminated:
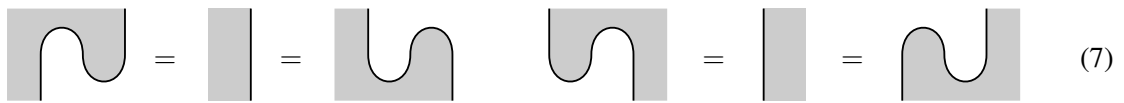
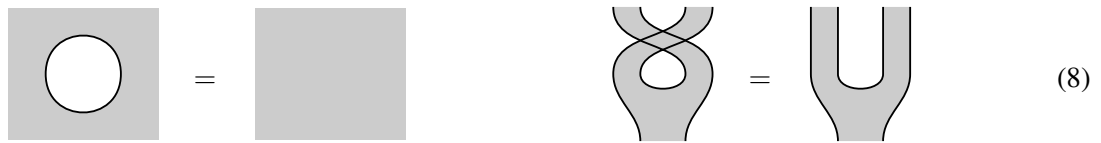
\def\aascale{0.7}
\begin{calign}
\def\quad{\hspace{0.3cm}}
\label{eq:top1}
\begin{aligned}
\begin{tikzpicture}[scale=\aascale,xscale=0.8, yscale=\ys]
\draw [use as bounding box, draw=none] (-0.5,0) rectangle (2.3,2);
\draw [white] (-0.5,0) to (3.1,2);
\draw [fill=\fillClight, draw=none] (-0.5,0) to (0.3,0) to (0.3,1)
    to [out=up, in=up, looseness=2] (1.3,1)
    to [out=down, in=down, looseness=2] (2.3,1)
    to (2.3,2) to (-0.5,2);
\draw [thick] (0.3,0) to (0.3,1)
    to [out=up, in=up, looseness=2] (1.3,1)
    to [out=down, in=down, looseness=2] (2.3,1)
    to (2.3,2);
\end{tikzpicture}
\end{aligned}
\quad=\quad
\begin{aligned}
\begin{tikzpicture}[scale=\aascale,xscale=0.8, yscale=\ys]
\draw [use as bounding box, draw=none] (1,0) rectangle (0,2);
\draw [white] (0,0) to (2,2);
\draw [fill=\fillClight, draw=none] (0,0)
    to (1,0)
    to (1,2)
    to (0,2);
\draw [thick] (1,0) to (1,2);
\end{tikzpicture}
\end{aligned}
\quad=\quad
\begin{aligned}
\begin{tikzpicture}[scale=\aascale,xscale=0.8, yscale=\ys]
\draw [use as bounding box, draw=none] (-0.5,0) rectangle (2.3,-2);
\draw [white] (-0.5,0) to (3.1,-2);
\draw [fill=\fillClight, draw=none] (-0.5,0) to (0.3,0) to (0.3,-1)
    to [out=down, in=down, looseness=2] (1.3,-1)
    to [out=up, in=up, looseness=2] (2.3,-1)
    to (2.3,-2) to (-0.5,-2);
\draw [thick] (0.3,0) to (0.3,-1)
    to [out=down, in=down, looseness=2] (1.3,-1)
    to [out=up, in=up, looseness=2] (2.3,-1)
    to (2.3,-2);
\end{tikzpicture}
\end{aligned}
&
\begin{aligned}
\begin{tikzpicture}[scale=\aascale,xscale=0.8, yscale=\ys]
\draw [use as bounding box, draw=none] (-2.3,0) rectangle (0.5,2);
\draw [white] (0.5,0) to (-3.1,2);
\draw [fill=\fillClight, draw=none] (0.5,0) to (-0.3,0) to (-0.3,1)
    to [out=up, in=up, looseness=2] (-1.3,1)
    to [out=down, in=down, looseness=2] (-2.3,1)
    to (-2.3,2) to (0.5,2);
\draw [thick] (-0.3,0) to (-0.3,1)
    to [out=up, in=up, looseness=2] (-1.3,1)
    to [out=down, in=down, looseness=2] (-2.3,1)
    to (-2.3,2);
\end{tikzpicture}
\end{aligned}
\quad=\quad
\begin{aligned}
\begin{tikzpicture}[scale=\aascale,xscale=0.8, yscale=\ys]
\draw [use as bounding box, draw=none] (-1,0) rectangle (0,2);
\draw [white] (0,0) to (-2,2);
\draw [fill=\fillClight, draw=none] (0,0)
    to (-1,0)
    to (-1,2)
    to (-0,2);
\draw [thick] (-1,0) to (-1,2);
\end{tikzpicture}
\end{aligned}
\quad=\quad
\begin{aligned}
\begin{tikzpicture}[scale=\aascale,xscale=0.8, yscale=\ys]
\draw [use as bounding box, draw=none] (-2.3,0) rectangle (0.5,-2);
\draw [white] (0.5,0) to (-3.1,-2);
\draw [fill=\fillClight, draw=none] (0.5,0) to (-0.3,0) to (-0.3,-1)
    to [out=down, in=down, looseness=2] (-1.3,-1)
    to [out=up, in=up, looseness=2] (-2.3,-1)
    to (-2.3,-2) to (0.5,-2);
\draw [thick] (-0.3,0) to (-0.3,-1)
    to [out=down, in=down, looseness=2] (-1.3,-1)
    to [out=up, in=up, looseness=2] (-2.3,-1)
    to (-2.3,-2);
\end{tikzpicture}
\end{aligned}
\\
\label{eq:top2}
\begin{aligned}
\begin{tikzpicture}[yscale=\ys]
\draw [fill=\fillClight, draw=none] (0.5,0.5) rectangle (2.5,2.5);
\draw [fill=white, thick] (1,1.5)
    to [out=up, in=up, looseness=2] (2,1.5)
    to [out=down, in=down, looseness=2] (1,1.5);
\end{tikzpicture}
\end{aligned}
\quad=\quad
\begin{aligned}
\begin{tikzpicture}[yscale=\ys]
\draw [fill=\fillClight, draw=none] (0.5,0.5) rectangle (2.5,2.5);
\end{tikzpicture}
\end{aligned}
&
\begin{aligned}
\begin{tikzpicture}[scale=0.6, yscale=-0.85, yscale=\ys]
\draw [fill=\fillClight, draw=none] (0.4,0)
    to [out=up, in=down, out looseness=1.3] (2,2)
    to (2.6,2)
    to [out=down, in=up, out looseness=1.3] (1,0);
\draw [fill=\fillClight, draw=none] (2.6,0)
    to [out=up, in=down, out looseness=1.3] (1,2)
    to (0.4,2)
    to [out=down, in=up, out looseness=1.3] (2,0);
\draw [fill=\fillClight, draw=none] (0.4,2)
    to (1,2)
    to [out=up, in=up, looseness=1.5] (2,2)
    to (2.6,2)
    to [out=up, in=down] (2,4)
    to (1,4)
    to [out=down, in=up] (0.4,2);
\draw [thick] (0.4,0)
    to [out=up, in=down, out looseness=1.3] (2,2)
    to [out=up, in=up, looseness=1.5] (1,2)
    to [out=down, in=up, in looseness=1.3] (2.6,0);
\draw [thick] (1,0)
    to [out=up, in=down, in looseness=1.3] (2.6,2)
    to [out=up, in=down] (2,4);
\draw [thick] (2,0)
    to [out=up, in=down, in looseness=1.3] (0.4,2)
    to [out=up, in=down] (1,4);
\end{tikzpicture}
\end{aligned}
\quad=\quad
\begin{aligned}
\begin{tikzpicture}[scale=0.6, yscale=-0.85, yscale=\ys]
\draw [fill=\fillClight, draw=none] (0.4,0)
    to [out=up, in=down, out looseness=1.5] (0.4,2)
    to (1,2)
    to [out=down, in=up, out looseness=1.5] (1,0);
\draw [fill=\fillClight, draw=none] (2.6,0)
    to [out=up, in=down, out looseness=1.5] (2.6,2)
    to (2,2)
    to [out=down, in=up, out looseness=1.5] (2,0);
\draw [fill=\fillClight, draw=none] (0.4,2)
    to (1,2)
    to [out=up, in=up, looseness=1.5] (2,2)
    to (2.6,2)
    to [out=up, in=down] (2,4)
    to (1,4) to [out=down, in=up] (0.4,2);
\draw [thick] (2.6,0)
    to [out=up, in=down, in looseness=1.5] (2.6,2)
    to [out=up, in=down] (2,4);
\draw [thick] (0.4,0)
    to [out=up, in=down, in looseness=1.5] (0.4,2)
    to [out=up, in=down] (1,4);
\draw [thick] (2,0) 
    to (2,2)
    to [out=up, in=up, looseness=1.5] (1,2)
    to (1,0);
\end{tikzpicture}
\end{aligned}
\end{calign}
\end{definition}

\noindent
Whenever we make use of the above graphical notation, it is understood that we are depicting an object with topological boundary in a symmetric monoidal dagger 2\-category.

\section{The $2[-]$ construction}
\label{sec:2construction}

This section introduces a construction that turns a monoidal category $\cat{C}$ into a 2\-category $\Two[\cat{C}]$, in such a way that $\Two[\CPs[\cat{C}]]$ has the appropriate structure to express the teleportation equation solely in terms of objects and morphisms. The idea is to adapt the well-known algebraic construction of rings, bimodules, and bimodule homomorphisms~\cite{hazewinkeletal:algebra}. In Section~\ref{sec:bimodules} we will see how our construction is defined, and in Section~\ref{sec:bimodboundary} we will see that objects in $2[\cat C]$ have a topological boundary in the sense of Definition~\ref{def:topologicalboundary}.

\subsection{Bimodules and composition}
\label{sec:bimodules}

\begin{definition}
  Let $(C,\tinymult[white dot],\tinyunit[white dot])$ and $(D,\tinymult[black dot],\tinyunit[black dot])$ be dagger Frobenius algebras in a dagger monoidal category. A \emph{dagger $C$-$D$-bimodule} is a morphism $\mathbf{M}$ satisfying:
  \begin{equation}\label{eq:dagger-bimodule}
    \begin{pic}[yscale=\ys]
      \node[morphism, minimum width=10mm] (u) at (0,0) {$\mathbf{M}$};
      \node[morphism, minimum width=10mm] (l) at (0,-1.1) {$\mathbf{M}$};
      \draw (u.north) to (0,.6) node[right] {$M$};
      \draw (u.south) to (l.north);
      \draw (l.south) to (0,-2) node[right] {$M$};
      \draw (l.-45) to[out=-90,in=90] (.6,-2) node[right] {$D$};
      \draw (l.-135) to[out=-90,in=90] (-.6,-2) node[right] {$C$};
      \draw (u.-45) to[out=-90,in=90] (1.1,-1) to (1.1,-2) node[right] {$D$};
      \draw (u.-135) to[out=-90,in=90] (-1.1,-1) to (-1.1,-2) node[right] {$C$};
    \end{pic}
    =
    \begin{pic}[yscale=\ys]
      \node[morphism, minimum width=10mm] (m) at (0,0) {$\mathbf{M}$};
      \node[white dot] (l) at (-.75,-1) {};
      \node[black dot] (r) at (.75,-1) {};
      \draw (m.north) to (0,.6) node[right] {$M$};
      \draw (m.south) to (0,-2) node[right] {$M$};
      \draw (m.-45) to[out=-90,in=90] (r);
      \draw (m.-135) to[out=-90,in=90] (l);
      \draw (l) to[out=-150,in=90] (-1.1,-2) node[left] {$C$};
      \draw (l) to[out=-30,in=90] (-.4,-2) node[left] {$C$};
      \draw (r) to[out=-150,in=90] (.4,-2) node[right] {$D$};
      \draw (r) to[out=-30,in=90] (1.1,-2) node[right] {$D$};
    \end{pic}
    \qquad
    \begin{pic}[yscale=\ys]
      \draw (0,.6) node[right] {$M$} to (0,-2) node[right] {$M$};
    \end{pic}
    \!\!\!=\,\,\,
    \begin{pic}[yscale=\ys]
      \node[morphism, minimum width=10mm] (m) at (0,0) {$\mathbf{M}$};
      \draw (m.north) to (0,.6) node[right] {$M$};
      \draw (m.south) to (0,-2) node[right] {$M$};
      \draw (m.-135) to[out=-90,in=90] (-.35,-1) node[white dot] {};
      \draw (m.-45) to[out=-90,in=90] (.35,-1) node[black dot] {};
    \end{pic}
    \qquad
    \begin{pic}[yscale=\ys]
      \node[morphism, minimum width=10mm] (m) at (0,-.7) {$\textbf{M}^\dag$};
      \draw (m.south) to (0,-2) node[right]{$M$};
      \draw (m.north) to  (0,.6) node[right]{$M$};
      \draw (m.45) to[out=90,in=-90] (.5,.6) node[right]{$D$};
      \draw (m.135) to[out=90,in=-90] (-.5,.6) node[right]{$C$};
    \end{pic}
    \!\!\!=\,\,\,
    \begin{pic}[yscale=\ys]
      \node[morphism, minimum width=10mm] (m) at (0,-.7) {$\textbf{M}$};
      \draw (m.south) to (0,-2) node[right]{$M$};
      \draw (m.north) to (0,.6) node[right]{$M$};
      \node[white dot] (l) at (-.6,-1.3) {};
      \node[black dot] (r) at (.6,-1.3) {};
      \draw (m.-45) to[out=-90,in=150] (r);
      \draw (m.-135) to[out=-90,in=30] (l);
      \draw (r) to (.6,-1.6) node[black dot] {};
      \draw (l) to (-.6,-1.6) node[white dot] {};
      \draw (l) to[out=150, in=-90, looseness=.5] (-1,.6) node[right] {$C$};
      \draw (r) to[out=30, in=-90, looseness=.5] (1,.6) node[right] {$D$};
    \end{pic}
  \end{equation}
  We also call the object $M$ the bimodule, and the map $\mathbf{M}$ its \emph{action}, and write $\mathbf{M}_{\tinydot[black dot]} = \mathbf{M} \smash{(\tinyid \; \tinyid {\raisebox{-2pt}{\tinyunit[black dot]}})}$ and ${}_{\tinydot[white dot]}\mathbf{M} = \mathbf{M} \smash{{\raisebox{-2pt}{(\tinyunit[white dot]}} \; \tinyid \; \tinyid)}$.
  A \emph{homomorphism} of dagger $C$-$D$-bimodules is a morphism $f \colon M \to M'$ that respects that actions by satisfying $f \mathbf{M} = \mathbf{M'} (\id[C] \otimes f \otimes \id[D])$.
  % A \emph{homomorphism} of dagger $C$-$D$-bimodules is a morphism $f$ satisfying:
  % \begin{equation}\label{eq:bimodule-homomorphism}
  %   \begin{pic}
  %     \node[morphism, minimum width=10mm] (m) at (0,0) {$\mathbf{M}$};
  %     \node[morphism] (f) at (0,1) {$f$};
  %     \draw (f.south) to (m.north);
  %     \draw (m.south) to (0,-1) node[right] {$M$};
  %     \draw (m.-45) to[out=-90,in=90] (.5,-1) node[right] {$D$};
  %     \draw (m.-135) to[out=-90,in=90] (-.5,-1) node[right] {$C$};
  %     \draw (f.north) to (0,2) node[right] {$M$};
  %   \end{pic}
  %   =
  %   \begin{pic}
  %     \node[morphism, minimum width=10mm] (m) at (0,1.3) {$\mathbf{M'}$};
  %     \node[morphism] (f) at (0,0) {$f$};
  %     \draw (f.north) to (m.south);
  %     \draw (f.south) to (0,-1) node[right] {$M$};
  %     \draw (m.-45) to[out=-90,in=90] (.6,0) to (.6,-1) node[right] {$D$};
  %     \draw (m.-135) to[out=-90,in=90] (-.6,0) to (-.6,-1) node[right] {$C$};
  %     \draw (m.north) to (0,2) node[right] {$M$};
  %   \end{pic}
  % \end{equation}
\end{definition}

If $M$ is a $C$-$D$-bimodule, and $N$ is a $D$-$E$-bimodule, the standard algebraic construction of \emph{tensor product} gives a $C$-$E$-bimodule $M \otimes_D N$; see~\cite[Section~4.5]{hazewinkeletal:algebra}. It is constructed by forcing the right $D$\-action on $M$ and the left $D$-action on $N$ to cooperate. More precisely, it is the coequalizer of the two morphisms $M \otimes D \otimes N \to M \otimes N$ induced by the two $D$-actions. 

One way to guarantee the existence of such a coequalizer  is to require that some morphisms have a sensible notion of \textit{image}, as in the following definition and lemma. Recall that an endomorphism $p \colon A \to A$ is a \emph{dagger idempotent} when $p^2=p=p^\dag$. A dagger idempotent $p$ \emph{splits} when $p=ii^\dag$ and $i^\dag i=\id$ for some morphism $i$, called the \emph{image} of $p$. Split idempotents are a special case of dagger coequalizers~\cite{selinger:idempotents}: a dagger idempotent $p \colon A \to A$ splits if and only if $p$ and $\id[A]$ have a dagger coequalizer~$i^\dag$.

\begin{definition}
  A dagger monoidal category \emph{has dagger Frobenius images} when for all classical structures $(C,\tinymult[white dot],\tinyunit[white dot])$, $(D,\tinymult[gray dot],\tinyunit[gray dot])$, $(E,\tinymult[black dot],\tinyunit[black dot])$, for all $C$-$D$-bimodules $\mathbf{M}$ and all $D$-$E$-bimodules $\textbf{N}$, the following dagger idempotent splits:
  \begin{equation}\label{eq:frobeniusimage}
    \begin{pic}[yscale=\ys]
      \draw [use as bounding box, draw=none] (-1.5,-1) rectangle (1.5,1);
      \node[morphism, minimum width=10mm] (m) at (-.75,.2) {$\mathbf{M}$};
      \node[morphism, minimum width=10mm] (n) at (.75,.2) {$\mathbf{N}$};
      \draw (m.north) to (-.75,1) node[right] {$M$};
      \draw (n.north) to (.75,1) node[right] {$N$};
      \draw (m.south) to (-.75,-1) node[right] {$M$};
      \draw (n.south) to (.75,-1) node[right] {$N$};
      \draw (m.-135) to[out=-90,in=90] (-1.1,-.6) node[white dot] {};
      \draw (n.-45) to[out=-90,in=90] (1.1,-.6) node[black dot] {};
      \draw (m.-45) to[out=-90,in=150] (0,-.4) node[gray dot] {} to[out=30,in=-90] (n.-135);
      \draw (0,-.5) to (0,-.75) node[gray dot] {};
    \end{pic}
  \end{equation}
  Notice that this morphism is indeed dagger idempotent by~\eqref{eq:dagger-bimodule}.
\end{definition}

We denote the image of~\eqref{eq:frobeniusimage} by $i \colon M \tinydot[gray dot] N \to M \otimes N$. It is a dagger $C$-$E$-bimodule:
\begin{equation}\label{eq:bimodule-tensor}
  \begin{pic}[yscale=\ys]
    \node[morphism, minimum width=10mm] (mn) at (0,0) {$\mathbf{M} \;\;\; \mathbf{N}$};
    \node[gray dot] at (0.05,0) {};
    \draw (mn.north) to (0,1.8) node[right] {$M \tinydot[gray dot] N$};
    \draw (mn.south) to (0,-2) node[right] {$M \tinydot[gray dot] N$};
    \draw (mn.-45) to[out=-90,in=90] (1,-2) node[right] {$E$};
    \draw (mn.-135) to[out=-90,in=90] (-1,-2) node[right] {$C$};
  \end{pic}
  \qquad := \qquad
  \begin{pic}[yscale=\ys]
      \node[morphism] (i) at (0,-1.3) {$i$};
      \node[morphism] (id) at (0,1.2) {$i^\dag$};
      \node[morphism, minimum width=10mm] (m) at (-.75,.2) {$\mathbf{M}$};
      \node[morphism, minimum width=10mm] (n) at (.75,.2) {$\mathbf{N}$};
      \draw (m.north) to[out=90,in=-90] (id.-135); 
      \draw (n.north) to[out=90,in=-90] (id.-45);
      \draw (m.south) to[out=-90,in=90] (i.135);
      \draw (n.south) to[out=-90,in=90] (i.45);
      \draw (m.-45) to[out=-90,in=150] (0,-.4) node[gray dot] {} to[out=30,in=-90] (n.-135);
      \draw (0,-.5) to (0,-.75) node[gray dot] {};
      \draw (m.-135) to[out=-90,in=90] (-1.25,-2) node[right] {$C$};
      \draw (n.-45) to[out=-90,in=90] (1.25,-2) node[right] {$E$};
      \draw (id.north) to (0,1.8) node[right] {$M \tinydot[gray dot] N$};
      \draw (i.south) to (0,-2) node[right] {$M \tinydot[gray dot] N$};
  \end{pic}
\end{equation}

\begin{lemma}\label{lem:coequalizer}
  If~\eqref{eq:frobeniusimage} splits with image $i$, then $i^\dag$ is a coequalizer of ${}_{\tinydot[white dot]}\mathbf{M} \otimes \id[N]$ and $\id[M] \otimes \mathbf{N}_{\tinydot[black dot]}$.
\end{lemma}
\begin{proof}
  Observe that $i^\dag ({}_{\tinydot[white dot]}\mathbf{M} \otimes \id[N]) = i^\dag (\id[M] \otimes \mathbf{N}_{\tinydot[black dot]})$ because $(\mathbf{M} \tinydot[gray dot] \mathbf{N}) ({}_{\tinydot[white dot]}\mathbf{M} \otimes \id[N]) = (\id[M] \otimes \mathbf{N}_{\tinydot[black dot]})({}_{\tinydot[white dot]}\mathbf{M} \otimes \id[N])$:
  \[
    \begin{pic}[yscale=\ys]
      \node[morphism, minimum width=10mm] (m) at (-.75,.2) {$\mathbf{M}$};
      \node[morphism, minimum width=10mm] (n) at (.75,.2) {$\mathbf{N}$};
      \node[morphism, minimum width=10mm] (l) at (-.75, -1) {$\mathbf{M}$};
      \draw (m.north) to (-.75,.8) node[right] {$M$};
      \draw (n.north) to (.75,.8) node[right] {$N$};
      \draw (m.south) to (l.north);
      \draw (n.south) to (.75,-2) node[right] {$N$};
      \draw (m.-135) to[out=-90,in=90] (-1.1,-.4) node[white dot] {};
      \draw (n.-45) to[out=-90,in=90] (1.1,-.4) node[black dot] {};
      \draw (m.-45) to[out=-90,in=150] (0,-.4) node[gray dot] {} to[out=30,in=-90] (n.-135);
      \draw (0,-.5) to (0,-.75) node[gray dot] {};
      \draw (l.south) to (-.75,-2) node[right] {$M$};
      \draw (l.-45) to[out=-90,in=90] (0,-2) node[right] {$D$};
      \draw (l.-135) to[out=-90,in=90] (-1.1,-1.6) node[white dot] {};
    \end{pic}
    \quad = \quad
    \begin{pic}[yscale=\ys]
      \node[morphism, minimum width=10mm] (m) at (-.75,.2) {$\mathbf{M}$};
      \node[morphism, minimum width=10mm] (n) at (.75,.2) {$\mathbf{N}$};
      \draw (m.north) to (-.75,.8) node[right] {$M$};
      \draw (n.north) to (.75,.8) node[right] {$N$};
      \draw (m.south) to (-.75,-2) node[right] {$M$};
      \draw (n.south) to (.75,-2) node[right] {$N$};
      \draw (m.-135) to[out=-90,in=90] (-1.1,-.4) node[white dot] {};
      \draw (n.-45) to[out=-90,in=90] (1.1,-.4) node[black dot] {};
      \draw (m.-45) to[out=-90,in=150] (0,-1) node[gray dot] {} to[out=30,in=-90] (n.-135);
      \draw (0,-1) to (0,-2) node[right] {$D$};
    \end{pic}
    \quad = \quad
    \begin{pic}[yscale=\ys]
      \node[morphism, minimum width=10mm] (m) at (-.75,.2) {$\mathbf{M}$};
      \node[morphism, minimum width=10mm] (n) at (.75,.2) {$\mathbf{N}$};
      \node[morphism, minimum width=10mm] (l) at (.75, -1) {$\mathbf{N}$};
      \draw (m.north) to (-.75,.8) node[right] {$M$};
      \draw (n.north) to (.75,.8) node[right] {$N$};
      \draw (n.south) to (l.north);
      \draw (m.south) to (-.75,-2) node[right] {$M$};
      \draw (m.-135) to[out=-90,in=90] (-1.1,-.4) node[white dot] {};
      \draw (n.-45) to[out=-90,in=90] (1.1,-.4) node[black dot] {};
      \draw (m.-45) to[out=-90,in=150] (0,-.4) node[gray dot] {} to[out=30,in=-90] (n.-135);
      \draw (0,-.5) to (0,-.75) node[gray dot] {};
      \draw (l.south) to (.75,-2) node[right] {$N$};
      \draw (l.-135) to[out=-90,in=90] (0,-2) node[right] {$D$};
      \draw (l.-45) to[out=-90,in=90] (1.1,-1.6) node[black dot] {};
    \end{pic}
  \]
  Suppose that $f ({}_{\tinydot[white dot]}\mathbf{M} \otimes \id[N]) = f (\id[M] \otimes \mathbf{N}_{\tinydot[black dot]})$. Then $f$ factors through $i^\dag$ as $f=fii^\dag$:
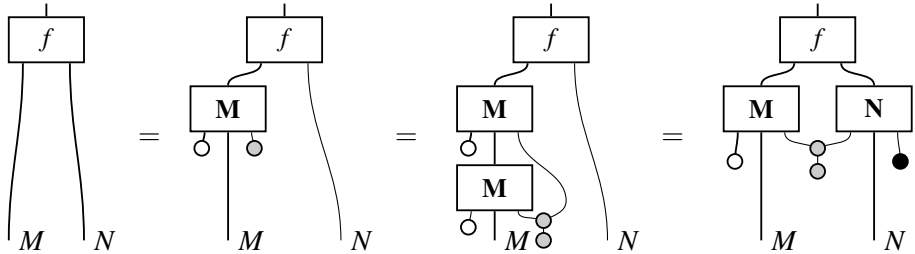
  \[
    \begin{pic}[yscale=\ys]
      \node[morphism, minimum width=10mm] (f) at (0,1.25) {$f$};
      \draw (f.north) to (0,1.8);
      \draw (f.-45) to[out=-90,in=90] (.5,-1.8) node[right] {$N$};
      \draw (f.-135) to[out=-90,in=90] (-.5,-1.8) node[right] {$M$};
    \end{pic}
    \; = \,\,\,\,\,\,
    \begin{pic}[yscale=\ys]
      \node[morphism, minimum width=10mm] (f) at (0,1.25) {$f$};
      \node[morphism, minimum width=10mm] (m) at (-.75,.2) {$\mathbf{M}$};
      \draw (f.north) to (0,1.8);
      \draw (m.north) to[out=90,in=-90] (f.-135);
      \draw (m.south) to (-.75,-1.8) node[right] {$M$};
      \draw (m.-135) to[out=-90,in=90] (-1.1,-.4) node[white dot] {};
      \draw (m.-45) to[out=-90,in=90] (-.4,-.4) node[gray dot] {};
      \draw (f.-45) to[out=-90,in=90] (.75,-1.8) node[right] {$N$};
    \end{pic}
    \; = \quad
    \begin{pic}[yscale=\ys]
      \node[morphism, minimum width=10mm] (f) at (0,1.25) {$f$};
      \node[morphism, minimum width=10mm] (m) at (-.75,.2) {$\mathbf{M}$};
      \node[morphism, minimum width=10mm] (m2) at (-.75,-1) {$\mathbf{M}$};
      \draw (f.north) to (0,1.8);
      \draw (m.north) to[out=90,in=-90] (f.-135);
      \draw (m.south) to (m2.north);
      \draw (m2.south) to (-.75,-1.8) node[right] {$M$};
      \draw (m.-135) to[out=-90,in=90] (-1.1,-.4) node[white dot] {};
      \draw (m2.-135) to[out=-90,in=90] (-1.1,-1.6) node[white dot] {};
      \draw (f.-45) to[out=-90,in=90] (.75,-1.8) node[right] {$N$};
      \draw (m2.-45) to[out=-90, in=150] (-.1,-1.5) node[gray dot] {} to[out=30,in=-90] (.2,-1.2) to[out=90,in=-90] (m.-45);
      \draw (-.1,-1.5) to (-.1,-1.8) node[gray dot] {};
    \end{pic}
    \; = \quad
    \begin{pic}[yscale=\ys]
      \node[morphism, minimum width=10mm] (f) at (0,1.25) {$f$};
      \node[morphism, minimum width=10mm] (m) at (-.75,.2) {$\mathbf{M}$};
      \node[morphism, minimum width=10mm] (n) at (.75,.2) {$\mathbf{N}$};
      \draw (f.north) to (0,1.8);
      \draw (m.north) to[out=90,in=-90] (f.-135);
      \draw (n.north) to[out=90,in=-90] (f.-45);
      \draw (m.south) to (-.75,-1.8) node[right] {$M$};
      \draw (n.south) to (.75,-1.8) node[right] {$N$};
      \draw (m.-135) to[out=-90,in=90] (-1.1,-.6) node[white dot] {};
      \draw (n.-45) to[out=-90,in=90] (1.1,-.6) node[black dot] {};
      \draw (m.-45) to[out=-90,in=150] (0,-.4) node[gray dot] {} to[out=30,in=-90] (n.-135);
      \draw (0,-.5) to (0,-.75) node[gray dot] {};
    \end{pic}
  \]
  This mediating map is unique: if $f=mi^\dag$, then $m=mi^\dag i=fi$.
\end{proof}

We can use this technique to re-prove many standard results about bimodules in a graphical way, such as the following simple result, that will be useful later.

\begin{lemma}\label{lem:identitybimodule}
  For any dagger Frobenius algebra $(A,\tinymult[white dot],\tinyunit[white dot])$, the identity $A$-$A$-bimodule is $\smash{{\hspace{-5pt}\ensuremath{\begin{pic}[scale=0.4,string,yscale=-1]
      \node (0) at (0,0) {};
      \node[white dot, inner sep=1.5pt] (1) at (0,0.55) {};
      \node (2) at (-0.5,1) {};
      \node (3) at (0.5,1) {};
      \draw (0.center) to (1.center);
      \draw (1.center) to [out=left, in=down, out looseness=1.5] (2.center);
      \draw (1.center) to [out=right, in=down, out looseness=1.5] (3.center);
      \draw (1.center) to (0,1);
      \end{pic}}\hspace{-3pt}}}$.
  \qed
\end{lemma}

With this preparation we can now define our main construction.
\begin{proposition}\label{prop:two}
  If $\cat{C}$ is a dagger monoidal category that has dagger Frobenius images, then the following data define a symmetric monoidal (weak) 2\-category $\Two[\cat{C}]$:
  \begin{itemize}
  \setlength\itemsep{0pt}
    \item objects are classical structures in $\cat{C}$;
    \item 1\-morphisms are dagger bimodules; 
      the identity 1\-morphism on $(C,\tinymult[gray dot],\tinyunit[gray dot])$ is 
      $\smash{{\hspace{-5pt}\ensuremath{\begin{pic}[scale=0.4,string,yscale=-1]
      \node (0) at (0,0) {};
      \node[gray dot, inner sep=1.5pt] (1) at (0,0.55) {};
      \node (2) at (-0.5,1) {};
      \node (3) at (0.5,1) {};
      \draw (0.center) to (1.center);
      \draw (1.center) to [out=left, in=down, out looseness=1.5] (2.center);
      \draw (1.center) to [out=right, in=down, out looseness=1.5] (3.center);
      \draw (1.center) to (0,1);
      \end{pic}}\hspace{-3pt}}}$;
    \item 2\-morphisms are dagger bimodule homomorphisms; 
    \item horizontal composition of 1\-morphisms is given by~\eqref{eq:bimodule-tensor};
    \item horizontal composition of 2\-morphisms follows from the universal property of Lemma~\ref{lem:coequalizer};
    \item monoidal structure is inherited from $\cat{C}$.
  \end{itemize}
\end{proposition}
  \noindent
  More precisely, the horizontal composition of 2\-morphisms $f \colon \mathbf{M} \to \mathbf{M'}$ and $g \colon \mathbf{N} \to \mathbf{N'}$ is the unique arrow making the following diagram commute:
  \[\begin{pic}[xscale=4,yscale=1.66, thin]
    \node (tl) at (-.3,1) {$M \otimes D \otimes N$};
    \node (bl) at (-.3,0) {$M' \otimes D \otimes N'$};
    \node (t) at (1,1) {$M \otimes N$};
    \node (b) at (1,0) {$M' \otimes N'$};
    \node (tr) at (1.7,1) {$M \tinydot[gray dot] N$};
    \node (br) at (1.7,0) {$M' \tinydot[gray dot] N'$};
    \draw[->] (tl) to node[left] {$(f \otimes \id[D] \otimes g)$} (bl);
    \draw[->] (t) to node[right] {$f \otimes g$} (b);
    \draw[->, dashed] (tr) to node[right] {$f \tinydot[gray dot] g$} (br);
    \draw[->] (t) to node[above] {$i^\dag$} (tr);
    \draw[->] (b) to node[below] {$i'^\dag$} (br);
    \draw[->] ([yshift=1.5pt]tl.east) to node[above] {${}_{\tinydot[white dot]}\mathbf{M} \otimes \id[N]$} ([yshift=1.5pt]t.west);
    \draw[->] ([yshift=-1.5pt]tl.east) to node[below] {$\id[M] \otimes \mathbf{N}_{\tinydot[black dot]}$} ([yshift=-1.5pt]t.west);
    \draw[->] ([yshift=1.5pt]bl.east) to node[above] {${}_{\tinydot[white dot]}\mathbf{M'} \otimes \id[N']$} ([yshift=1.5pt]b.west);
    \draw[->] ([yshift=-1.5pt]bl.east) to node[below] {$\id[M'] \otimes \mathbf{N'}_{\tinydot[black dot]}$} ([yshift=-1.5pt]b.west);
  \end{pic}\]
\begin{proof}
  For verification that these data indeed satisfy all the conditions required of a weak 2\-category, see~\cite{wester:mscthesis}. Verifying monoidality is a huge exercise that nevertheless seems straightforward enough.
\end{proof}

\noindent
We end this subsection by listing some properties of the $\Two[-]$--construction; for proofs we refer to~\cite{wester:mscthesis}.
\begin{itemize}
  \setlength\itemsep{0pt}
  \item If $\cat{C}$ has a dagger, so does $\Two[\cat{C}]$.%, in the sense that all hom-categories $\Two[\cat{C}](C,D)$ have daggers in a way that is compatible with horizontal composition.
  \item If $\cat{C}$ is compact, so is $\Two[\cat{C}]$: 1\-morphisms have duals that are both left and right adjoint.
  \item If $\cat{C}$ has dagger biproducts, so do all hom-categories of $\Two[\cat{C}]$.
  \item The scalars of $\Two[\cat{C}]$ correspond to $\cat{C}$: there is an isomorphism $\Two[\cat{C}](I,I) \cong \cat{C}$ of categories.
\end{itemize}

\subsection{Topological boundaries}
\label{sec:bimodboundary}

We now show that objects of $2[\cat C]$ have topological boundaries in the sense of Definition~\ref{def:topologicalboundary}.
\begin{definition}
  The \emph{boundaries} of a special dagger Frobenius algebra $(A, \tinymult[white dot], \tinyunit[white dot])$ in \cat{C} are canonical bimodules $(A,\tinymult[white dot],\tinyunit[white dot]) \xto {\mathbf L} (I,\lambda_I,\id[I])$ and $(I,\lambda_I, \id[I]) \xto {\mathbf R} (A, \tinymult[white dot],\tinyunit[white dot])$ induced by the multiplication map $\tinymult[white dot]$:
\begin{align}
\begin{aligned}
\begin{tikzpicture}[yscale=\ys]
\node (m) [morphism, minimum width=10mm] at (0,0) {$\mathbf{L}$};
\draw [dashed] (m.-45) to (m.-45)+(0,-0.75);
\draw (m.-90) to (m.-90)+(0,-0.75);
\draw (m.-135) to (m.-135)+(0,-0.75);
\draw (m.90) to (m.90)+(0,0.75);
\end{tikzpicture}
\end{aligned}
\quad&:=\quad
\begin{aligned}
\begin{tikzpicture}[yscale=\ys]
\draw (0,-1.05) to (0,-0.5) to [out=up, in=0] (-0.25,0) to (-0.25,1.05);
\draw (-0.5,-1.05) to (-0.5,-0.5) to [out=up, in=180] (-0.25,0) node [whitedot] {};
\draw [dashed] (0.5,-1.05) to [out=up, in=0] (0,-0.5);
\end{tikzpicture}
\end{aligned}
&
\begin{aligned}
\begin{tikzpicture}[yscale=\ys]
\node (m) [morphism, minimum width=10mm] at (0,0) {$\mathbf{R}$};
\draw (m.-45) to (m.-45)+(0,-0.75);
\draw (m.-90) to (m.-90)+(0,-0.75);
\draw [dashed] (m.-135) to (m.-135)+(0,-0.75);
\draw (m.90) to (m.90)+(0,0.75);
\end{tikzpicture}
\end{aligned}
\quad&:=\quad
\begin{aligned}
\begin{tikzpicture}[yscale=\ys]
\draw (0,-1.05) to (0,-0.5) to [out=90, in=180] (0.25,0) to (0.25,1.05);
\draw (0.5,-1.05) to (0.5,-0.5) to [out=90, in=0] (0.25,0) node [whitedot] {};
\draw [dashed] (-0.5,-1.05) to [out=90, in=180] (0,-0.5);
\end{tikzpicture}
\end{aligned}
\end{align}
The dashed lines indicate the monoidal unit object; we will typically omit these from now on. These boundaries are depicted as lines that bound solid regions, as shown in the diagrams~\eqref{eq:firsttopboundary}.
\end{definition}
\begin{lemma}
\label{lem:rlcomposite}
For a special dagger Frobenius algebra $(A, \tinymult[white dot], \tinyunit[white dot])$, the composite bimodule $$(I,\lambda_I,\id[I]) \xto {\mathbf R} (A,\tinymult[white dot],\tinyunit[white dot]) \xto{\mathbf L} (I,\lambda_I,\id[I])$$ is isomorphic to the object $A$.
\end{lemma}
\begin{proof}
By Lemma~\ref{lem:coequalizer}, we must find the dagger splitting of the left-hand diagram below:
\begin{equation}
\begin{aligned}
\begin{tikzpicture}[yscale=\ys]
\draw (0,0) to (0,1) to [out=up, in=180] (0.25,1.25) node [whitedot] {} to (0.25,1.75);
\draw (0.25,1.25) to [out=0, in=180] (0.75,0.8) node [whitedot] {} to [out=0, in=180] (1.25,1.25) node [whitedot] {};
\draw (0.75,0.75) to (.75,.45) node [whitedot] {};
\draw (1.25,1.25) to [out=0, in=up] (1.5,1) to (1.5,0);
\draw (1.25,1.25) to (1.25,1.75);
\end{tikzpicture}
\end{aligned}
\quad=\quad
\begin{aligned}
\begin{tikzpicture}[yscale=\ys]
\draw (0,0) to [out=up, in=180] (0.5,0.5) node [whitedot] {} to (0.5,1.25) node [whitedot] {} to [out=180, in=down] (0,1.75);
\draw (0.5,0.5) to [out=0, in=up] (1,0);
\draw (0.5,1.25) to [out=0, in=down] (1,1.75);
\end{tikzpicture}
\end{aligned}
\end{equation}
By the dagger Frobenius axioms it equals the right-hand diagram. But the dagger specialness axiom makes this a dagger splitting via the object $A$, so the object $A$ gives the composite of the bimodules.
\end{proof}

\begin{lemma}
  The boundaries of a special dagger Frobenius monoid $(A, \tinymult[white dot], \tinyunit[white dot])$ in \cat C can be equipped with data~\eqref{eq:secondtopboundary} and \eqref{eq:lasttopboundary} satisfying equations~\eqref{eq:top1} and \eqref{eq:top2} as follows:
\begin{calign}
\begin{aligned}
\begin{tikzpicture}
    \draw [fill=\fillClight, draw=none] (0.2,-0.5)
        to (0.5,-0.5)
        to [out=down,in=150] (1,-1)
        to [out=30, in=down] (1.5,-0.5)
        to (1.8,-0.5)
        to (1.8,-1.5)
        to (0.2,-1.5)
        to (0.2,-0.5);
    \draw [thick] (0.5,-0.5) node [above] {$\mathbf L\vphantom{{}_A}$}
        to [out=down, in=150] (1,-1)
        to [out=30, in=down] (1.5,-0.5) node [above] {$\mathbf R\vphantom{{}_A}$};
\draw [thick] (1,-1) node [whitedot] {} to (1,-1.5) node [below] {$\id [A]\vphantom{\mathbf L}$};
\end{tikzpicture}
\end{aligned}
&
\begin{aligned}
\begin{tikzpicture}[yscale=-1]
    \draw [fill=\fillClight, draw=none] (0.2,-0.5)
        to (0.5,-0.5)
        to [out=down,in=150] (1,-1)
        to [out=30, in=down] (1.5,-0.5)
        to (1.8,-0.5)
        to (1.8,-1.5)
        to (0.2,-1.5)
        to (0.2,-0.5);
    \draw [thick] (0.5,-0.5) node [below] {$\mathbf L\vphantom{{}_A}$}
        to [out=down, in=150] (1,-1)
        to [out=30, in=down] (1.5,-0.5) node [below] {$\mathbf R$};
\draw [thick] (1,-1) node [whitedot] {} to (1,-1.5) node [above] {$\id [A]\vphantom{\mathbf L}$};
\end{tikzpicture}
\end{aligned}
&\begin{aligned}
\begin{tikzpicture}
    \draw [thick, fill=\fillClight] (0.5,-0.5) node [above] {$\mathbf L\vphantom{{}_A}$}
        to [out=down, in=150] (1,-1)
        to [out=30, in=down] (1.5,-0.5) node [above] {$\mathbf R\vphantom{{}_A}$};
\draw [thick, dashed] (1,-1) node [whitedot] {} to (1,-1.5) node [below] {$\id [I]\vphantom{\mathbf L}$};
\end{tikzpicture}
\end{aligned}
&
\begin{aligned}
\begin{tikzpicture}[yscale=-1]
    \draw [thick, fill=\fillClight] (0.5,-0.5) node [below] {$\mathbf L\vphantom{{}_A}$}
        to [out=down, in=150] (1,-1)
        to [out=30, in=down] (1.5,-0.5) node [below] {$\mathbf R\vphantom{{}_A}$};
\draw [thick, dashed] (1,-1) node [whitedot] {} to (1,-1.5) node [above] {$\id [I]\vphantom{\mathbf L}$};
\end{tikzpicture}
\end{aligned}
\\
\nonumber
\tinycomult[white dot] \colon A \to A \otimes A
&
\tinymult[white dot] \colon A \otimes A \to A
&
\tinyunit[white dot] \colon I \to A
&
\tinycounit[white dot] \colon A \to I
\end{calign}
\end{lemma}
\noindent
Note that we are relying on Lemmas~\ref{lem:identitybimodule} and~\ref{lem:rlcomposite} for these definitions to make sense.
\begin{proof}
  Equations~\eqref{eq:top1} follow immediately from the (co)unit equation for a dagger Frobenius algebra. Equations~\eqref{eq:top2} follow immediately from specialness and commutativity.
\end{proof}

\section{The case of Hilbert spaces}
\label{sec:2cpshilb}

This section discusses $\Two[\CPs[\cat{FHilb}]]$. 
We show that a substantial portion is well-defined, which we characterize in concrete terms: it consists of natural numbers, matrices of finite-dimensional C*-algebras, and matrices of completely positive maps. Thus this is completely analogous to the case of $\Two[\cat{FHilb}]$, which is equivalent to the 2\-category of 2\-Hilbert spaces that consists of natural numbers, matrices of Hilbert spaces, and matrices of linear maps~\cite{baez:twohilbertspaces,wester:mscthesis}. The difficulty of establishing that $2[\CPs[\cat{FHilb}]]$ is well-defined in general arises because $\CPs[\cat{FHilb}]$ does not have good completeness properties.

\begin{lemma}
  Not all coequalizers in the category $\CPs[\cat{FHilb}]$ are split epimorphisms.
\end{lemma}
\begin{proof}
  Due to the dagger we may equivalently show that not all equalizers split. Suppose the completely positive maps $f=(\begin{smallmatrix} 1 & 1 & 0 & 0\end{smallmatrix})$ and $g=(\begin{smallmatrix} 0 & 0 & 1 & 1 \end{smallmatrix}) \colon \C^4 \to \C$ had an equalizer $e \colon A \to \C^4$ in $\CPs[\cat{FHilb}]$. Then $fe=ge$, so $e$ factors through the equalizer of $f$ and $g$ in $\cat{FHilb}$:
  \[\begin{pic}[xscale=3,yscale=1.5]
    \node (l) at (0,0) {$\C^3$};
    \node (m) at (1,0) {$\C^4$};
    \node (r) at (2,0) {$\C$};
    \node (b) at (0,-1) {$A$};
    \draw[->] (b) to node[right, below]{$e$} (m);
    \draw[->] ([yshift=2pt]m.east) to node[above]{$f$} ([yshift=2pt]r.west);
    \draw[->] ([yshift=-2pt]m.east) to node[below]{$g$} ([yshift=-2pt]r.west);
    \draw[->, dashed] (b) to node[left]{$m$} (l);
    \draw[->] (l) to node[above]{$\left(\begin{smallmatrix} -1 & 1 & 1 \\ 1 & 0 & 0 \\ 0 & 1 & 0 \\ 0 & 0 & 1 \end{smallmatrix}\right)$} (m);
  \end{pic}\]
  We show below that the function $e$ is injective\footnote{We thank Narutaka Ozawa for this observation.}. Then the linear map $m$ is injective, and so $\dim(A) \leq 3$. It follows that $A$ must be a commutative C*-algebra (as 2-by-2 matrices already have dimension 4).

  Suppose $e(a)=0$ with $a=x+iy$ for self-adjoint $x,y \in A$. Then $e(x)=e(y)=0$ because positive maps preserve adjoints~\cite[page~2]{stormer:positive}. Say $x=x_+-x_-$ for positive $x_+,x_- \in A$; then $e(x_+)=e(x_-)$. But the completely positive maps $h_\pm \colon \C \to A$ defined by $h_\pm(1)=x_\pm$ satisfy $eh_+=eh_-$. So $h_+=h_-$ since $e$ is monic, whence $x=0$. Similarly $y=0$. So $\ker(e)=\{0\}$, and $e$ is injective.

  On the other hand, there are at least four completely positive maps $\C \to \C^4$, given by $x_1=(1,0,1,0)$, $x_2=(1,1,0,0)$, $x_3=(0,1,0,1)$, $x_4=(0,0,1,1)$, which satisfy $fx_i=gx_i$. No $x_i$ is a linear combination of the others with nonnegative coefficients. 
  If $d$ is a retraction of $e$, therefore none of $d x_i \in A$ is a linear combination of the others with nonnegative coefficients, as $e d x_i=x_i$. Moreover $d x_i \geq 0$ by completely positivity of $d$. But this contradicts $\dim(A)\leq 3$ as $A$ is commutative.
\end{proof}

It follows immediately that $\CPs[\cat{FHilb}]$ does not have dagger coequalizers.
The point is that there are nevertheless enough coequalizers for our purposes, as we show below.

\subsection{Analysis}

A subcollection of the objects in $\CPs[\cat{FHilb}]$ are classical structures $C=(A,\tinymult[white dot],\tinyunit[white dot])$ in \cat{FHilb}. Since the morphisms $\tinymult[white dot]$ and $\tinyunit[white dot]$ are completely positive with respect to this algebra structure, this also gives rise to an algebra $C' = (C,\tinymult[white dot],\tinyunit[white dot])$ in $\CPs[\cat{FHilb}]$. We call this a \emph{classical structure over itself}. Note that, up to isomorphism $C \cong \C^n$, such structures are just natural numbers. In this section, except for the last theorem, we restrict consideration to objects of $2[\CPs[\cat{FHilb}]]$ which are classical structures over themselves.

\begin{lemma}\label{lem:matrices-objects}
  There is a one-to-one correspondence between dagger bimodules on classical structures over themselves in $\CPs[\cat{FHilb}]$, and matrices of finite-dimensional C*-algebras.
\end{lemma}
\begin{proof}
  Let $\mathbf{M} \colon \C^m \otimes M \otimes \C^n \to M$ be a dagger $\C^m$-$\C^n$-bimodule in $\CPs[\cat{FHilb}]$, between classical structures over themselves. Then $M$ is a finite-dimensional C*-algebra, and $\mathbf{M}$ is a completely positive map. By \ref{thm:groupoid}, the units of $\C^m$ and $\C^n$ are given by the sum over the standard basis vectors $\ket{i}$ and $\ket{j}$ of $\C^m$ and $\C^n$, respectively. Set $p_{ij} = (\ket{i}\bra{i}) \otimes \id[M] \otimes (\ket{j}\bra{j})$; this is a completely positive dagger idempotent. Hence its image $M_{ij}=p_{ij}(M)$ is a finite-dimensional C*-algebra by a classic theorem of Choi and Effros~\cite{choieffros:injectivity}; see~\cite[Proposition~2.4]{heunenkissingerselinger:cpproj}. Thus the bimodule $\mathbf{M}$ gives rise to a matrix $(M_{ij})$ of finite-dimensional C*-algebras.

  Conversely, let $(M_{ij})$ be an $m$-by-$n$ matrix of finite-dimensional C*-algebras. Set $M=\bigoplus_{i,j} M_{ij}$, and define $\mathbf{M} \colon \C^m \otimes M \otimes C^n \to M$ by mapping $\ket{i} \otimes a \otimes \ket{j}$ to $1_{ij} \cdot a$, where $1_{ij}$ is the unit of $M_{ij}$. In other words, $\mathbf{M}(\ket{i} \otimes a \otimes \ket{j})$ is the projection of $a \in A$ onto the summand $M_{ij}$. This is a $*$\-homomorphism, and hence a completely positive map~\cite[Lemma~3.8]{coeckeheunenkissinger:cpstar}. To verify that it is a bimodule, we need to check equation~\eqref{eq:dagger-bimodule}. The first two equalities are easily verified, the third equality uses that $\mathbf{M}^\dag \colon M \to \C^m \otimes M \otimes \C^n$ maps $b \in M$ to $\sum_{i,j} \ket{i} \otimes (1_{ij}b) \otimes \ket{j}$.
  % \begin{align*}
  %   \braket{ \mathbf{M}( \ket{i_0} \otimes a \otimes \ket{j_0}) }{ b }_M
  %   & = \braket{ 1_{i_0j_0}a }{ b }_{\bigoplus_{i,j} M_{ij}} \\
  %   & = \sum_{i,j} \braket{ 1_{ij} 1_{i_0j_0} a }{ 1_{ij}b }_{M_{ij}} \\
  %   & = \braket{ 1_{i_0j_0}a }{ 1_{i_0j_0}b }_{M_{ij}} \\
  %   & = \Tr(a^* b 1_{i_0j_0}) \\
  %   & = \braket{ a }{ 1_{i_0j_0}b }_M \\
  %   & = \sum_{i,j} \braket{i_0}{i}_{\C^n} \braket{a}{1_{ij}b}_M \braket{j_0}{j}_{\C^n} \\
  %   & = \sum_{i,j} \braket{ ( \ket{i_0} \otimes a \otimes \ket{j_0}) }{ (\ket{i} \otimes 1_{ij}b \otimes \ket{j}) }_{\C^m \otimes M \otimes \C^n} \\
  %   & = \braket{ ( \ket{i_0} \otimes a \otimes \ket{j_0}) }{ \mathbf{M}^\dag(b) }_{\C^m \otimes M \otimes \C^n}.
  % \end{align*}
Hence these two constructions, which are inverse to each other, are well-defined.
\end{proof}

It follows that an important part of $\Two[\CPs[\cat{FHilb}]]$ is well-defined, which will be sufficient for our applications in Section~\ref{sec:applications} to quantum information and encryption.
\begin{proposition}\label{cor:restriction}
  Let $(C,\tinymult[white dot],\tinyunit[white dot])$, $(D,\tinymult[gray dot],\tinyunit[gray dot])$, $(E,\tinymult[black dot],\tinyunit[black dot])$ be classical structures over themselves in $\CPs[\cat{FHilb}]$, and let $\mathbf{M}$ and $\mathbf{N}$ be a $C$-$D$-bimodule and a $D$-$E$-bimodule. The idempotent~\eqref{eq:frobeniusimage} splits.
\end{proposition}
\begin{proof}
  Write $\ket{i}$, $\ket{j}$, $\ket{k}$ for the standard bases of $\C^l$, $\C^m$, $\C^n$. Let $\mathbf{M}$ be a dagger $\C^l$-$\C^m$-bimodule, and let $\mathbf{N}$ be a dagger $\C^m$-$\C^n$-bimodule in $\CPs[\cat{FHilb}]$. Then~\eqref{eq:frobeniusimage} maps $m \otimes n$ to $\sum_{i,j,k} \mathbf{M}(\ket{i} \otimes m \otimes \ket{j}) \otimes \mathbf{N}(\ket{j} \otimes n \otimes \ket{k})$. This morphism is a sum of orthogonal projections, and hence a projection itself. As in the proof of Lemma~\ref{lem:matrices-objects}, this means that it has a well-defined dagger image in $\CPs[\cat{FHilb}]$.
  The proof is finished by noticing that any classical structure in $\cat{FHilb}$ is isomorphic to the commutative C*-algebra $\C^n$ for some $n$.
\end{proof}

\begin{lemma}\label{lem:matrices-morphisms}
There is a one-to-one correspondence between homomorphisms of dagger bimodules between classical structures over themselves in $\CPs[\cat{FHilb}]$, and matrices of completely positive maps between finite-dimensional C*-algebras.
\end{lemma}
\begin{proof}
  Let $f \colon \mathbf{M} \to \mathbf{N}$ be a homomorphism of dagger $\C^m$-$\C^n$-bimodules in $\CPs[\cat{FHilb}]$. Write $\ket{i}$ and $\ket{j}$ for the standard bases of $\C^m$ and $\C^n$. According to the proof of Lemma~\ref{lem:matrices-objects}, let $p_{ij} \colon M \to M_{ij}$ and $q_{ij} \colon N \to N_{ij}$ be the completely positive maps implementing the biproduct decompositions $M=\bigoplus_{i,j} M_{ij}$ and $N=\bigoplus_{i,j} N_{ij}$. Then $f_{ij} = q_{ij} f \smash{p_{ij}^\dag} \colon M_{ij} \to N_{ij}$ is an $m$-by-$n$ matrix of completely positive maps. Conversely, let $(f_{ij})$ be an $m$-by-$n$ matrix of completely positive maps $f_{ij} \colon M_{ij} \to N_{ij}$. According to Lemma~\ref{lem:matrices-objects} we have to find a map $f \colon M \to N$ for $M=\bigoplus_{i,j} M_{ij}$ and $N=\bigoplus_{i,j} N_{ij}$. Just take $f=\bigoplus_{ij} f_{ij}$; this is well-defined because $\CPs[\cat{C}]$ inherits biproducts from $\cat{C}$~\cite[Theorem~3.2]{heunenkissingerselinger:cpproj}. We have to verify that this is a well-defined homomorphism of dagger bimodules:
  \begin{align*}
    f \mathbf{M} ( \ket{i_0} \otimes a \otimes \ket{j_0})
    & = f(1_{i_0j_0}a)
%    \\&
    = \textstyle \bigoplus_{i,j} f_{ij} (1_{i_0j_0}a)
%    \\&
    = f_{i_0j_0}(1_{i_0j_0}a)
    \\&
    = \textstyle 1_{i_0j_0} \bigoplus_{i,j} f_{ij}(1_{ij}a)
    \\&
    = 1_{i_0j_0} f(a)
%    \\&
    = \mathbf{N}( \ket{i_0} \otimes f(a) \otimes \ket{j_0} )
  \end{align*}
These two constructions are clearly inverse to each other.
\end{proof}

We can now characterize a well-defined part of $\Two[\CPs[\cat{FHilb}]]$.
\begin{theorem}
\label{thm:welldefined}
  The following full sub-2\-category is well-defined within $\Two[\CPs[\cat{FHilb}]]$:
  \begin{itemize}
  \setlength\itemsep{0pt}
    \item objects are natural numbers $m$;
    \item 1\-morphisms $m \to n$ are $m$-by-$n$ matrices $(M_{ij})$ of finite-dimensional C*-algebras;
    \item 2\-morphisms $(M_{ij}) \to (N_{ij})$ are $m$-by-$n$ matrices $(f_{ij})$ of completely positive maps;
    \item horizontal composition of 1\-morphisms is given by $(\bigoplus_j M_{ij} \otimes N_{jk})$;
    \item horizontal composition of 2\-morphisms is given by $(\bigoplus_j f_{ij} \otimes g_{jk})$;
    \item vertical composition of 2\-morphisms is given by $(g_{ij} f_{ij})$.
  \end{itemize}
\end{theorem}
\begin{proof}
  It suffices to show that the correspondences of Lemmas~\ref{lem:matrices-objects} and~\ref{lem:matrices-morphisms} turn the compositions of Proposition~\ref{prop:two} into the ones of the statement. Let $\mathbf{M}$ and $\mathbf{M'}$ be $\C^l$-$\C^m$-bimodules, and let $\mathbf{N}$ and $\mathbf{N'}$ be $\C^m$-$\C^n$-bimodules. These correspond to matrices of C*-algebras, where $M_{ij}$ is the image of $\ket{i}\bra{i} \otimes \id[M] \otimes \ket{j} \bra{j}$. Let $f \colon \mathbf{M} \to \mathbf{M'}$ and $g \colon \mathbf{N} \to \mathbf{N'}$ be bimodule homomorphisms. These correspond to matrices of completely positive maps $f_{ij} \colon M_{ij} \to M'_{ij}$ and $g_{ij} \colon N_{ij} \to N_{ij}'$. Now, by definition $M \tinydot[gray dot] N$ is the image of the map $\sum_{i,j,k} \mathbf{M}(\ket{i} \otimes [-] \otimes \ket{j}) \otimes \mathbf{N}(\ket{j} \otimes [-] \otimes \ket{k})$. But this is just $\bigoplus_{i,j,k} M_{ij} \otimes N_{jk}$. Similarly, horizontal composition of $f$ and $g$ corresponds to $(\bigoplus_j f_{ij} \otimes g_{jk})$. 
\end{proof}

In future work we would of course like to show that $\Two[\CPs[\cat{FHilb}]]$ is completely well-defined. The first task will be to characterize its objects up to isomorphism. We offer the following theorem, which generalizes~\cite[Corollary~3.10]{coeckeduncankissingerwang:mermin}, as evidence that this is a nontrivial question. Recall that a state $x \in C$ of a classical structure $(C,\tinymult[black dot],\tinyunit[black dot])$ in $\cat{FHilb}$ is \emph{copyable} when $\tinycomult[black dot](x) = x \otimes x$.

\begin{theorem}\label{thm:groupoid}
  Consider a classical structure $C$ in $\cat{FHilb}$ as an object of $\CPs[\cat{FHilb}]$.
  There is a one-to-one correspondence between dagger special Frobenius algebras on $C$ in $\CPs[\cat{FHilb}]$, and finite groupoids whose morphisms are the copyable states of $C$.
\end{theorem}
\begin{proof}
  Let $(\C^n,\tinymult[white dot],\tinyunit[white dot])$ be a dagger special Frobenius algebra on $\C^n$ in $\CPs[\cat{FHilb}]$.
  That is, it is a dagger special Frobenius algebra in $\cat{FHilb}$---\textit{i.e.}\ a finite-dimensional C*-algebra~\cite{vicary:quantumalgebras}---satisfying the extra condition that $\tinymult[white dot]$ and $\tinyunit[white dot]$ are completely positive maps. Since they are maps between commutative C*-algebras, saying that $\tinymult[white dot]$ and $\tinyunit[white dot]$ are completely positive is the same as saying that they are linear maps that preserve positive elements~\cite[Theorem~1.2.4]{stormer:positive}.
  Write $\tinymult[white dot]$ and $\tinyunit[white dot]$ as a matrix using the standard basis $\ket{i}$ of $\C^n$. Then all matrix entries $\bra{i}\tinymult[white dot]\ket{jk}$ and $\bra{i}\tinyunit[white dot]\ket{1}$ are nonnegative real numbers, and conversely, if all the matrix entries are nonnegative, then the linear maps $\tinymult[white dot]$ and $\tinyunit[white dot]$ certainly preserve positive elements.
  Thus $(\C^n,\tinymult[white dot],\tinyunit[white dot])$ is a dagger special Frobenius algebra in $\CPs[\cat{FHilb}]$ if and only if it is a C*-algebra whose multiplication and unit have nonnegative matrix entries on the standard basis $\ket{i}$ of $\C^n$.

  But then, by~\cite[Proposition~34]{abramskyheunen:hstar}, the matrix entries of $\tinymult[white dot]$ must in fact be either $0$ or $1$ (see also~\cite[Section~5.2]{coeckeheunenkissinger:cpstar}.)
  So we may equally think of the matrix of $\tinymult[white dot]$ as a morphism in the category $\cat{Rel}$ of sets and relations, where it still is a special dagger Frobenius algebra. Hence it encodes the multiplication of a groupoid whose arrows are the row indices $\ket{i}$~\cite{heunencontrerascattaneo}. 
  As units for a monoid are unique, also the matrix of $\tinyunit[white dot]$ must take values in $\{0,1\}$, and encode the identities of the groupoid. Finally, any classical structure $C$ in $\cat{FHilb}$ is isomorphic to $\C^n$ for some $n$, with the standard basis of $\C^n$ corresponding to the copyable states of $C$. Similarly, a $*$-isomorphism between classical structures in $\cat{Rel}$ corresponds to an isomorphism of groupoids~\cite[Theorem~19]{heunencontrerascattaneo}.
\end{proof}

\noindent
We leave open the interesting question of whether isomorphism between these objects in $2[\CPs[\cat{FHilb}]]$ (so-called \emph{Morita equivalence}) corresponds to \emph{equivalence} of groupoids.

\section{Applications}
\label{sec:applications}

We now consider applications to quantum information of the well-defined part of $2[\CPs[\cat{FHilb}]]$ constructed in Theorem~\ref{thm:welldefined}. We give an abstract 2\-categorical definition of \textit{measurement}, and show it recovers the ordinary notion positive operator--valued measure. We then analyze the 2\-categorical equation for quantum teleportation, and show that it has solutions in our 2\-category given by both encrypted communication and quantum teleportation. We then give a proof of a security property, which applies simultaneously to both types of solution.

\subsection{Measurement}

Earlier work on the 2\-categorical syntax for pure-state quantum theory~\cite{vicary:higherquantumtheory} demonstrated that a projective quantum measurement corresponds to a \textit{unitary} 2\-morphism which converts a local system into an extended system. Since our measurements in general are mixed, unitarity is not appropriate; instead we impose a counit-preservation condition.
\newcommand\tinymatrixmult{\smash{\ensuremath{
\begin{tikzpicture}[xscale=0.5, thick, scale=0.27]
\draw [use as bounding box, draw=none] (0,0.06) rectangle (3,1);
\draw (0,0) to [out=up, in=down] (1,1);
\draw (1,0) to [out=up, in=up] (2,0);
\draw (3,0) to [out=up, in=down] (2,1);
\end{tikzpicture}}}}
\newcommand\tinymatrixunit{\smash{\ensuremath{
\begin{tikzpicture}[xscale=0.5, thick, scale=0.27]
\draw [use as bounding box, draw=none] (1,0.06) rectangle (2,1);
\draw (1,1) to [out=down, in=down, looseness=2] (2,1);
\end{tikzpicture}}}}

\begin{definition}\label{def:measurement}
In $2[\CPs[\cat{FHilb}]]$, a \emph{measurement} is a counit-preserving 2\-morphism of type:
\begin{equation}
\begin{aligned}
\begin{tikzpicture}[xscale=0.7, thick, yscale=\ys]
\draw [fill=\fillClight] (-1,2) to [out=down, in=150] (0,1) to [out=30, in=down] (1,2);
\draw (0,0) node [below] {} to (0,1) node [circle, inner sep=0pt, minimum width=0.6cm, draw, fill=white] {$\mu$};
%\node at (0,1.6) {$\C ^n$};
\end{tikzpicture}
\end{aligned}
\end{equation}
\end{definition}

It is not ideal that we must modify the definition of a measurement in this way. The situation is analogous to the work in~\cite{stayvicary}, where measurements were required to be kernel-free. That requirement can be replaced with the more elegant unitarity condition~\cite{barvicary}. With further work we hope to show the same in the current setting, a task which is likely to require making use of a larger part of $2[\CPs[\cat{FHilb}]]$ than we have so-far shown to be well-defined.
However, Definition~\ref{def:measurement} elegantly captures precisely the desired notion, as we now show.

\begin{theorem}
Restricting to the part of $2[\CPs[\cat{FHilb}]]$ defined in Theorem~\ref{thm:welldefined}, measurements on matrix algebras are exactly positive operator-valued measures.
\end{theorem}
\begin{proof}
The 2\-morphism $\mu$ is a trace-preserving completely positive map from a matrix algebra to a classical structure. Its adjoint $\mu ^\dag$ is therefore a completely positive map out of a classical structure. Such a map is completely defined by its action on the $n$ copyable states of the classical structure, which must be sent to positive elements of $H \otimes H^*$. Thus $\mu ^\dag$ is defined by a family of $n$ positive operators $P_i : H \to H$.

The counit-preservation condition is given by the left-hand condition below:
\begin{equation}
\left(\,\,\,
\begin{aligned}
\begin{tikzpicture}[xscale=0.7, thick, yscale=\ys]
\draw [use as bounding box, draw=none] (-1,0) rectangle (1,2.6);
\draw [fill=\fillClight] (-1,2) to [out=down, in=150] (0,1) to [out=30, in=down] (1,2) to [out=up, in=up] (-1,2);
\draw (0,0) to (0,1) node [circle, inner sep=0pt, minimum width=0.6cm, draw, fill=white] {$\mu$};
\node at (0,1.9) {$\C ^n$};
\end{tikzpicture}
\end{aligned}
\quad=\quad
\begin{aligned}
\begin{tikzpicture}[xscale=0.7, thick, yscale=\ys]
\draw [use as bounding box, draw=none] (-0.5,0) rectangle (.5,2.6);
\draw (0,0) to (0,1) node [ground] {};
\end{tikzpicture}
\end{aligned}
\,\,\,
\right)
\qquad\stackrel{\dag}\Leftrightarrow\qquad
\left(\,\,\,
\begin{aligned}
\begin{tikzpicture}[xscale=0.7, thick, yscale=-1, yscale=\ys]
\draw [use as bounding box, draw=none] (-1,0) rectangle (1,2.6);
\draw [fill=\fillClight] (-1,2) to [out=down, in=150] (0,1) to [out=30, in=down] (1,2) to [out=up, in=up] (-1,2);
\draw (0,0) to (0,1) node [circle, inner sep=-4pt, minimum width=0.6cm, draw, fill=white] {$\mu ^\dag$};
\node at (0,1.9) {$\C ^n$};
\end{tikzpicture}
\end{aligned}
\quad=\quad
\begin{aligned}
\begin{tikzpicture}[xscale=0.7, thick, yscale=-1, yscale=\ys]
\draw [use as bounding box, draw=none] (-0.5,0) rectangle (.5,2.6);
\draw (0,0) to (0,1) node [ground, hflip] {};
\end{tikzpicture}
\end{aligned}
\,\,\,
\right)
\end{equation}
On the right-hand side we take the adjoint of this condition. We use the `earth' symbol to represent the counit of a matrix algebra, which is just the trace map, following previous work~\cite{coeckeperdrix}. The second equation says precisely that $\sum _i P_i = \id[H]$, which is exactly the condition for the family of positive operators $P_i$ to define a positive operator--valued measurement.
\end{proof}

\subsection{Unification of quantum teleportation and classical encrypted communication}

\begin{definition}%[Teleportation]
In a symmetric monoidal 2\-category containing an object with a topological boundary, \emph{teleportation} is a solution to the following equation with $\mu$ a measurement and $\nu$ unitary:
\begin{equation}
\label{eq:teleportation}
\begin{aligned}
\begin{tikzpicture}[thick, yscale=\ys]
\draw [use as bounding box, draw=none] (0,0) rectangle (2.5,3);
\draw (0,0) to [out=up, in=-150] (1,1)
  to [out=-30, in=left, in looseness=1] (2.0,0.3)
  to [out=right, in=down] (2.5,1)
  to [out=up, in=-30] (2,2)
  to [out=30, in=down] (2.5,3);
\draw[fill=\fillClight](0,3) to [out=down, in=150] (1,1)
  to [out=30, in=-150] (2,2)
  to [out=150, in=down] (1.5,3);
\node at (1,1) [circle, inner sep=0pt, minimum width=0.6cm, draw, fill=white] {$\mu$};
\node at (2,2) [circle, inner sep=0pt, minimum width=0.6cm, draw, fill=white] {$\nu$};
\end{tikzpicture}
\end{aligned}
\qquad=\qquad
\begin{aligned}
\begin{tikzpicture}[thick, yscale=\ys]
\draw [use as bounding box, draw=none] (0,0) rectangle (2.5,3);
\draw[fill=\fillClight](0,3) to (0,2.4)
  to [out=down, in=down, looseness=2] (1.5,2.4)
  to (1.5,3);
\draw (0,0) to [out=up, in=down, in looseness=2] (2.5,3);
\end{tikzpicture}
\end{aligned}
\end{equation}
\end{definition}

\noindent
Note that this definition relies on our earlier Definition~\ref{def:measurement} of a measurement.

\begin{theorem}
\label{thm:teleportationsolutions}
When the nontrivial region is labelled by a discrete groupoid, solutions to the teleportation equation in $2[\CPs[\cat{FHilb}]]$ can be obtained as follows:
\begin{enumerate}
\item when the incoming system is a classical structure, by implementations of classical encrypted communication using a one-time pad;
\item when the incoming system is a matrix algebra, by implementations of quantum teleportation.
\end{enumerate}
\end{theorem}
\begin{proof}
We can only give a sketch here. It is already established separately that both classical encrypted communication via a one-time pad~\cite{stayvicary} and quantum teleportation~\cite{vicary:higherquantumtheory} can be characterized exactly as solutions to this equation, in \cat{2Rel} and \cat{2Hilb} respectively. Both families of solutions can be embedded into $2[\CPs[\cat{FHilb}]]$ in an appropriate fashion.
\end{proof}

\noindent
It is an interesting open question whether these are the only solutions, or whether solutions exist which somehow \textit{mix} the encryption and teleportation aspects.

\subsection{Security of teleportation}

In both quantum teleportation and classical encrypted communication with a one-time pad, it is true that if you throw away the second half of the cryptographic resource---the entangled state or the secret key, respectively---all information about the message is lost. An abstract proof of this has already been given in the 2\-categorical setup for the case of encrypted communication~\cite{stayvicary}. We now give a general proof that applies simultaneously to quantum teleportation and encrypted communication.
\begin{theorem}
For any solution to the teleportation equation~\eqref{eq:teleportation}, destroying the second half of the shared resource is equivalent to destroying the original message:
\begin{equation}
\label{eq:teleportationsecurity}
\begin{aligned}
\begin{tikzpicture}[thick, yscale=\ys]
\draw [use as bounding box, draw=none] (0,0) rectangle (2.5,3);
\draw (0.25,0) to [out=up, in=-150] (1,1)
  to [out=-30, in=left, in looseness=1] (2.0,0.3)
  to [out=right, in=down] (2.5,1)
  to (2.5,2.0) node [ground] {};
\draw[fill=\fillClight] (0.25,3) to (0.25,2.5)
  to [out=down, in=150] (1,1)
  to [out=30, in=down] (1.75,2.5) to (1.75,3);
\node at (1,1) [circle, inner sep=0pt, minimum width=0.6cm, draw, fill=white] {$\mu$};
\end{tikzpicture}
\end{aligned}
\qquad=\qquad
\begin{aligned}
\begin{tikzpicture}[thick, yscale=\ys]
\draw [use as bounding box, draw=none] (0,0) rectangle (2.5,3);
\draw[fill=\fillClight](0,3) to (0,2.8)
  to [out=down, in=down, looseness=2] (1.5,2.8)
  to (1.5,3);
\draw (0.75,0) to [out=up, in=down, in looseness=2] (0.75,1) node [ground] {};
\end{tikzpicture}
\end{aligned}
\end{equation}
\end{theorem}
\begin{proof}
Adjoin a trace map to the final system on both sides of the teleportation equation~\eqref{eq:teleportation}. The map $\nu$ is a family of invertible completely positive maps by Lemma~\ref{lem:matrices-morphisms}, and thus is necessarily trace-preserving~\cite[Theorem~3.3]{cariello}; the left-hand side therefore simplifies, giving equation~\eqref{eq:teleportationsecurity}.
\end{proof}

\setlength\bibsep{0pt}
\setlength\bibindent{0pt}
\setlength\bibhang{-10pt}
 \def\bibfont{\footnotesize}
\bibliographystyle{eptcs}
\bibliography{2cp}

\end{document}